\documentclass[aps,pra,superscriptaddress,nofootinbib,onecolumn, notitlepage]{revtex4-1}

\usepackage{physics}
\usepackage{xcolor}
\usepackage{amssymb}
\usepackage{mathtools}
\usepackage{bbm}
\usepackage{graphicx} 
\usepackage{amsmath}
\usepackage{amsfonts}
\usepackage{braket}
\usepackage{fontenc}
\usepackage{mathrsfs}
\usepackage{verbatim}
\usepackage{dsfont}
\usepackage{csquotes}
\usepackage{comment}
\usepackage{amsthm}
\usepackage[colorlinks,citecolor=blue,urlcolor=blue,bookmarks=false,hypertexnames=true]{hyperref}
\usepackage{ulem}
\usepackage{bm}
\usepackage{enumitem}

\theoremstyle{plain}
\newtheorem{thm}{Theorem}[]

\newtheorem*{cor*}{Corollary}

\newtheorem*{thm*}{Theorem}
\newtheorem{lem}[thm]{Lemma}

\newtheorem{cor}[thm]{Corollary}
\newtheorem{defn}[thm]{Definition}
\theoremstyle{remark}

\newcommand{\rc}[1]{{\color{red} #1}}

\newcommand{\sumab}[2]{\underset{#1}{\overset{#2}{\sum}}}
\newcommand{\suma}[1]{\underset{#1}{\sum}}

\newcommand{\bigcupab}[2]{\underset{#1}{\overset{#2}{\bigcup}}}

\renewcommand{\vec}[1]{{\bm{#1}}}

\begin{document} 

\title{On the entropy growth of sums of iid discrete random variables}

\date{\today}

\author{Riccardo Castellano}

\affiliation{Department of Applied Physics, University of Geneva, 1211 Geneva, Switzerland}

\author{Pavel Sekatski}
\affiliation{Department of Applied Physics, University of Geneva, 1211 Geneva, Switzerland}

\begin{abstract}
 We derive an asymptotic lower bound on the Shannon entropy $H$ of sums of $N$ arbitrary iid discrete random variables. The derived bound $H \geq \frac{r(X)}{2}\log(N) + {\it cst}$ is given in terms of the  incommensurability rank $r(X)$ of the random variable -- a positive integer quantity that we introduce.  The derivation does not rely on central limit theorems, but builds upon the known expressions of the asymptotic entropy of the multinomial distribution and sums of iid lattice random variables, which correspond to the case $r(X)=1$.
\end{abstract}
\maketitle

\section{Introduction and background}

\label{app: iid}

Let $\{X_{i}\}_{i \in \mathbb{N}}$ be a sequence of independent and identically distributed (iid) random variables that take values in a  countable  subset of real numbers $\chi\subset \mathbb{R}$, and have a finite variance ${\rm Var}(X)$. We are interested to lower bound the Shannon entropy $H$ of the sum $T_{N}:= \sumab{i=1}{N}X_{i}$ with an expression of the form
\begin{equation}\label{eq: desiderata}
    H(T_{N})\geq \frac{r(\chi)}{2} \log(N\lambda(X))+o(1).
\end{equation}

Historically, it is the case of continuos random variables that attracted more attention~\cite{Linnik,barron1986entropy,artstein2004solution,madiman2007generalized}. Here, it is the {\it differential} Shannon entropy $h(T_N)\to \frac{1}{2}\log(N)+\frac{1}{2}\log(2\pi e {\rm Var}(X))$ that converges~\cite{barron1986entropy} to that of a Gaussian random variable  from below~\cite{artstein2004solution,madiman2007generalized}.\\

In the case of discrete random variables, $T_N$ is supported on a discrete set, which plays a crucial role for the calculation of the Shannon entropy $H(T_N)$. A bound of the form~\eqref{eq: desiderata} is known for the particular case of {\it lattice} random variables, which we now define. 

\begin{defn}\label{defn:lattice}
$X$ is a {\bf lattice random variable}, if it takes values in a {\bf lattice set} $\chi$, i.e. such that 
\begin{align}\label{app: lattice}
    \chi \subset \big \{ h \, n + a : n \in \mathbb{Z}\big\}
\end{align}
for  some $h, a \in \mathbb{R}$. The {\bf maximal span} $h(X)=h(\chi)$ is the largest real number $h$ for which Eq.~\eqref{app: lattice} can be fulfilled. 
\end{defn}
Notice that random variables that only take one or two values ($|\chi| \leq 2$) are always a lattice. 
\begin{thm}\cite{EntropySumIID,takano1987convergence}\label{Th:LatticeEntropy}
    Let $X$ be a lattice random variable. Then in the $N \rightarrow \infty$ limit
    \begin{equation}\label{eq: greek thrm}
        H(T_{N})= \frac{1}{2}\log(N)+\frac{1}{2}\log(2\pi e \frac{{\rm Var}(X)}{h(\chi)^{2}})+o(1).
    \end{equation}
\end{thm}

 Upon a shift by $a$ and a rescaling by $\frac{1}{h}$, irrelevant for the entropy, the values of a lattice random variable are integer $n_i \in \mathds{Z}$. Intuitively, if $\chi$ is not an affine transformation of a set of integers, but contains numbers that are somehow {\it incommensurable},  there can only be fewer different combinations of outcomes that sum up to the same value of $T_{N}$, making the entropy $H(T_N)$ larger.  The goal of this paper is to formalize this intuition and derive a generalized version of Theorem~\ref{Th:LatticeEntropy} that applies for all discrete random variables.\\

As a last preliminary result we recall that the asymptotic entropy of a multinomial distribution has a simple expression.

\begin{thm}\label{th:multinomEntropy} \cite{MultinomialEntropy}
    Let $\bm{n}(N,\bm{p})\in \mathbb{N}^d$ be a multinomial random variable describing $N$ trials with $d$ different outcomes of probabilities $\bm{p}:=(p_{1},..p_{d})$. In the large $N$ limit its entropy is given by 
\begin{align}
    & H(\bm{n}(N,\bm{p})) = \frac{d-1}{2}\log(2\pi e N \sqrt[d-1]{p_{1}..p_{d}})+o(1) 
\end{align}
\end{thm}

\section{Incommensurable partitions and the first generalization}

First, we need to define some notion of ``incommensurability'' of the set of values $\chi$ of the random variables $X_i$, which would quantify how far we are from the case of a lattice random variable. 
Consider a {\bf partition} of a set 
$\chi=\bigcupab{j=1}{k} \chi_j$ 
into $k$ disjoint subsets $\{\chi_{j} \}_{j=1..k}$. For each of these sets let 
\begin{equation}
\mathcal{S}_{m}(\chi_j) := \left\{ \sum_{x_{i}\in \chi_{j}}n_{i} x_{i} : 
n_i \in \mathbb{N}, \sum_i n_i = m
\right \}
\end{equation}
be the set of numbers that can be obtained by summing $m$, not necessarily different, numbers from $\chi_{j}$. For convenience we also introduce the following notation 
\begin{equation}
    \mathbb{N}^{k}_{N}:=\left\{ \bm{n}\in \mathbb{N}^{k} : \sum_{i=1}^{k} n_{i}=N \right\},
\end{equation}
for the set of natural $k$-vectors $\bm n$ whose elements sum up to $N$ (with $0 \in \mathbb{N}$).  With its help, we define the following property.

\begin{defn}\label{def: incom}
A collection of disjoint sets $\{\chi_{j} \}_{j=1..k}$ of 
real numbers (e.g. a partition of $\chi$) is called \textbf{incommensurable} (or $k$-incommensurable) if 
\begin{enumerate}
    \item All $\chi_{j}$ are non-empty and $\chi_j \neq \{0\}$.\\
\item
For all $N\in \mathbb{N}$, $\bm{m},\bm{n}\in \mathbb{N}^{k}_{N}$ and $ \vec{y}= (y_1,\dots,y_k)$, $\vec{y}'= (y_1',\dots,y_k')$ with $y_j\in S_{m_{j}}(\chi_{j})$ and $y_j' \in \mathcal{S}_{n_{j}}(\chi_j)$  we have 
\begin{equation}\label{eq: incom def}
\sum_{j=1}^k y_j = \sum_{j=1}^k y_j' \implies y_j =y_j' \quad \forall j \in [1..k].
\end{equation}
\end{enumerate}

\end{defn}

In words, for a random variable $X$ taking values in $\chi$, which admits a $k$-\textit{incommensurable} partition, any possible value $T_N = t$ can be \textit{uniquely} decomposed as $t=y_1+\dots + y_k$, where each $y_j$ is a sum of numbers sampled from $\chi_j$ and there are $N$ sampled numbers in total. This property implies a bijection between the outcome of $T_{N}$ and the $k$ outcomes of the following random variables 
\begin{equation}
Y_{j}^{(N)}:=\sumab{i=1}{N} X_{i} \, \mathbf{1}_{\chi_{j}}(X_{i}), \qquad \text{where} \qquad 
\mathbf{1}_{\chi_{j}}(x) = 
\begin{cases}
1 & x \in  \chi_{j} \\
0 & \text{otherwise}
\end{cases}
\end{equation}
are indicator functions verifying if $X_i$ takes value in $\chi_j$, and each $Y_{j}^{(N)}$ only sums the values drawn from $\chi_j$.\\

Some examples and observations on how to construct an incommensurable partition of $\chi$ can be found in App. \ref{IncommensurabilityExamples}.
In all of the examples the sets $\chi$ admit incommensurable partitions in sets $\chi_j$ which are lattices. For such cases, the entropy of $T_N$ can be bounded with the help theorems \ref{Th:LatticeEntropy} and \ref{th:multinomEntropy}, indeed we now prove the following bound.

\begin{thm}
    \label{Th:EntropyForIncomm}
    Let $X_1,\dots, X_N$ be a collection of iid discrete random variables taking values in the set $\chi$, which admits an $k$-incommensurable partition $\{\chi_j\}_{j=1,\dots,k}$, with  $p_j = {\rm Pr}(X_i \in \chi_j)$, and such that  each $\chi_j$ is a lattice, and $\chi_1,\dots \chi_s$ are the only sets with a single element. Then, the following bounds hold
    \begin{align}
            s=k : \,\,  H\left(\sum_{i=1}^N X_i \right) & = 
          \frac{k-1}{2}\log\left(2\pi e N \sqrt[k-1]{p_{1}p_{2}..p_{k}}\right)+o(1) \\ \label{eq: Ther6 gen lb}
          s\leq k-1 :\,\,   H\left(\sum_{i=1}^N X_i \right) & \geq
          \frac{k}{2} \log(2\pi e N)+ \frac{1}{2}\log\left( p_{1}\dots p_s (1- \sum_{j=1}^s p_j) \right) +\sum_{j=s+1}^k\frac{1}{2}\log( p_j \frac{{\rm Var}(\tilde X^{(j)})}{h(\chi_j)^{2}})+ o(1) \\
          s\leq k-2 :\, \, H\left(\sum_{i=1}^N X_i \right) & \leq  \frac{k+(k-s-1)}{2} \log(2\pi e N)+ \frac{1}{2}\log\left( p_{1}\dots p_s (1- \sum_{j=1}^s p_j) \right) \nonumber\\
          & +\sum_{j=s+1}^k\frac{1}{2}\log( p_j^2 \frac{{\rm Var}(\tilde X^{(j)})}{h(\chi_j)^{2}})+ o(1)
    \end{align}
    where the conditional random variables $\tilde X^{(j)}$ have probability distribution fixed as ${\rm Pr}(\tilde X^{(j)}=x) ={\rm Pr}(X=x| X \in \chi_j)$ and $h(\chi_j)$ is the maximal span of $\chi_j$. In addition, for $s=k-1$ the expression \eqref{eq: Ther6 gen lb} is an equality.
\end{thm}
\begin{proof}

Consider the following random variables
\begin{align}
        & n_N^{(j)}:=\sum_{i=1}^{N} \mathbf{1}_{\chi_{j}}(X_{i})  \qquad \text{with} \qquad \bm 1_{\chi_j}(x) = 
    \begin{cases}
        1 & x \in \chi_j \\
        0 & x \notin \chi_j 
    \end{cases}\\
        &Y_N^{(j)}:= \sum_{i=1}^{N} X_{i} \, \mathbf{1}_{\chi_{j}}(X_{i})
    \end{align}
    where $n_N^{(j)}$ counts the number of times the variables $X_i$ took values in each set $\chi_j$, and $Y_N^{(j)}$ is the sum of these values, with $T_N =\sumab{i=1}{N} X_i = \sumab{j=1}{k} Y_N^{(j)}$. Since the sets $\chi_j$ are incommensurable, we known that each value $T_N =t$ admits a unique decomposition $t = \suma{j} y_j$ as a sum of values of the variable $Y_N^{(j)}$. Hence, there is a bijection between the random variable $T_N$ and the vector of random variables $\bm Y_N = (Y_N^{(1)},\dots, Y_N^{(k)})$, implying
    \begin{equation}
        H(T_N) = H(\bm Y_N).
    \end{equation}

Recall that only the first $s\leq k$ sets $\chi_j$ contain a single element $x_j\neq 0$ (by definition~\ref{def: incom}). These sets are special, as for them the identity $Y_N^{(j)} = x_j n_N^{(j)}$ gives a bijection between $Y_N^{(j)} $ and $n_N^{(j)}$. Using the chain rule we can thus write
\begin{align}
    H(\bm Y_N) &= H(Y_N^{(1)},\dots,Y_N^{(s)}) + H(Y_N^{(s+1)},\dots,Y_N^{(k)}|Y_N^{(1)},\dots,Y_N^{(s)})\\
   &= H(n_N^{(1)},\dots,n_N^{(s)}) + H(Y_N^{(s+1)},\dots,Y_N^{(k)}|n_N^{(1)},\dots,n_N^{(s)}). \label{eq: chanin rule}
\end{align}

Here, the first term $H(n_N^{(1)},\dots,n_N^{(s)})$ can be readily computed. To do so we will distinguish two cases. For $s=k=|\chi|$, we have $n_N^{(1)}+\dots+n_N^{(s)}=N$ and the vector $(n_N^{(1)},\dots,n_N^{(s)})$ follows a multinomial distribution with $s=k=|\chi|$ outcomes of probabilities $p_1,\dots p_s$. By Theorem~\ref{th:multinomEntropy}
we find 
\begin{equation}
    H(\bm Y_N) = H(n_N^{(1)},\dots,n_N^{(s)}) = \frac{k-1}{2}\log\left(2\pi e N \sqrt[k-1]{p_{1}p_{2}..p_{k}}\right)+o(1) 
\end{equation}
proving the theorem for the special case. 
For $s<k$, $n_N^{(1)}+\dots+n_N^{(s)}\neq N$ in general and the vector of random variables $(n_N^{(1)},\dots,n_N^{(s)})$ follows the multinomial distribution with N trials and $s+1$ events of probabilities $p_1,\dots, p_s, 1- \sumab{j=1}{s} p_j$, where the last event corresponds to $X_{j}\notin \chi_{1},..\chi_{s}$ . Again, using Theorem~\ref{th:multinomEntropy} we find
\begin{align}
    H(n_N^{(1)},\dots,n_N^{(s)}) & = \frac{s}{2}\log\left(2\pi e N \big( p_{1}\dots p_s (1- \sum_{j=1}^s p_j)\big)^{1/s} \right)-o(1) \\
    &= \frac{s}{2}\log( 2 \pi e N)  + \frac{1}{2}\log\left( p_{1}\dots p_s (1- \sum_{j=1}^s p_j) \right)-o(1)
\end{align}

To bound the second term in Eq.~\eqref{eq: chanin rule} we note that the knowledge of all $n_N^{(1)},\dots,n_N^{(s)}$ does not reveal more information on the remaining variables $Y_N^{(j)}$ (with $j\geq s+1$) than the knowledge of their sum $\sumab{j=1}{s} n_N^{(j)}$, therefore
\begin{align}
H\left(Y_N^{(s+1)},\dots,Y_N^{(k)}|n_N^{(1)},\dots,n_N^{(s)}\right) &=
H\left(Y_N^{(s+1)},\dots,Y_N^{(k)}|\sum_{j=1}^s n_N^{(j)}\right)
\end{align}
Since the total number of trials is fixed, knowing $\sumab{j=1}{s} n_N^{(j)}$ is equivalent to knowing $\sumab{j=s+1}{k} n_N^{(j)}= N-\sumab{j=1}{s} n_N^{(j)}$. 

Next, we separately discuss the lower and upper-bounds on $H(Y_N^{(s+1)},\dots,Y_N^{(k)}|\sumab{j=s+1}{k} n_N^{(j)})$. For the lower bound, by data processing inequality we conclude that the conditional entropy could only be lower if we knew all the individual values of $n_N^{(j)}$, hence
 \begin{equation}
     H\left(Y_N^{(s+1)},\dots,Y_N^{(k)}|\sum_{j=s+1}^k n_N^{(j)}\right) \geq  H\left(Y_N^{(s+1)},\dots,Y_N^{(k)}|n_N^{(s+1)},\dots, n_N^{(k)}\right). \label{eq: entropy lb}
 \end{equation}
For the upper bound, we essentially use the processing inequality with the chain rule again to get
\begin{align}
     H\left(Y_N^{(s+1)},\dots,Y_N^{(k)}|\sum_{j=s+1}^k n_N^{(j)}\right) &\leq  H\left(Y_N^{(s+1)},\dots,Y_N^{(k)},n_N^{(s+1)},\dots, n_N^{(k)}|\sum_{j=s+1}^k n_N^{(j)}\right)\label{eq: entropy ub} \\
     &=   H\left(Y_N^{(s+1)},\dots,Y_N^{(k)}|n_N^{(s+1)},\dots, n_N^{(k)}\right) + H\left(n_N^{(s+1)},\dots, n_N^{(k)}|\sum_{j=s+1}^k n_N^{(j)}\right) \label{eq: entropy ub last}
\end{align}
 Remark that both bounds~(\ref{eq: entropy lb},\ref{eq: entropy ub}) are tight for $s=k-1$ where $n_N^{(k)}=\sumab{j=s+1}{k} n_N^{(j)}$.

 Let us first compute the last term in the upper bound~\eqref{eq: entropy ub last}. To do so introduce $q=\sumab{j=s+1}{k} p_j$, and note that for $n=\sumab{j=s+1}{k} n_N^{(j)}$ the numbers $\left(n_N^{(s+1)},\dots, n_N^{(k)}\right)$ follow a multinomial distribution with $k-s$ outcomes of probability $\frac{p_j}{q}$. Therefore, by  Theorem~\ref{th:multinomEntropy} we get
\begin{align}
    H\left(n_N^{(s+1)},\dots, n_N^{(k)}|\sum_{j=s+1}^k n_N^{(j)}\right) &= \sum_{n=0}^N \binom{N}{n}\, q^n (1-q)^{N-n}\, H\left(n_N^{(s+1)},\dots, n_N^{(k)}|\sum_{j=s+1}^k n_N^{(j)}=n\right)\\
     &= \sum_{n=0}^N \binom{N}{n}\, q^n (1-q)^{N-n}
     \frac{k-s-1}{2}\log\left(2\pi e n \big(\frac{p_{s+1}}{q}\dots \frac{p_k}{q} \big)^{1/(k-s-1)} \right)+o(1)\\
     & = \frac{k-s-1}{2}\log\left(2\pi e  N\right) + \frac{1-\delta_{1,k-s-1}}{2}\log\left(p_{s+1\dots p_k}\right) + o(1),
\end{align}
where we used the fact that for the binomial distribution $\mathds{E}[\log(n)]= \log (q N)+o(1)$ in the assymptotic limit.
 
Next, compute the expression on the rhs of \eqref{eq: entropy lb}, which also appears in the upper bound. In fact,  when the total numbers of time $n_N^{(j)}$ the outcomes of $X_i$ fall into each set $\chi_j$ are known, the variables $Y_N^{(j)}$ become independent. Therefore,
\begin{equation}
   H(Y_N^{(s+1)},\dots,Y_N^{(k)}|n_N^{(s+1)},\dots, n_N^{(k)}) = \sum_{j=s+1}^k H(Y_N^{(j)}|n_N^{(j)}),
\end{equation}
and it remains to bound the individual entropies. To do so note that each $n_N^{(j)}$ follows a binomial distribution so 
\begin{align}
    H(Y_N^{(j)}|n_N^{(j)}) &= \sum_{j=0}^N \binom{N}{n}\, p_j^n (1-p_j)^{N-n} H(Y_N^{(j)}|n_N^{(j)}=n) \\
    &= \sum_{j=0}^N \binom{N}{n}\, p_j^n (1-p_j)^{N-n} H\left(\sum_{i=1}^n \tilde X_i^{(j)}\right),
\end{align}
were $\tilde X_i^{(j)} \in \chi_j$ with ${\rm Pr}(\tilde X_i^{(j)}= x) ={\rm Pr}( X_i= x|X_i\in \chi_j)=\frac{{\rm Pr}( X_i= x)}{p_j}$. Since each $\chi_j$ is a lattice by assumptions, using Theorem~\ref{Th:LatticeEntropy}
    we find
\begin{align}
    H(Y_N^{(j)}|n_N^{(j)}) &= \frac{1}{2}\log(N p_j)+\frac{1}{2}\log(2\pi e \frac{{\rm Var}(\tilde X^{(j)})}{h(\chi_j)^{2}})-o(1)\\
    & = \frac{1}{2}\log(2\pi e N)+\frac{1}{2}\log( p_j \frac{{\rm Var}(\tilde X^{(j)})}{h(\chi_j)^{2}})-o(1),
\end{align}

Combining all the terms together we thus find
\begin{align}
    H(T_N) \geq \frac{k}{2} \log(2\pi e N)+ \frac{1}{2}\log\left( p_{1}\dots p_s \left(1- \sum_{j=1}^s p_j \right) \right) +\sum_{j=s+1}^k\frac{1}{2}\log( p_j \frac{{\rm Var}(\tilde X^{(j)})}{h(\chi_j)^{2}})- o(1),
\end{align}
which is tight for $s=k-1$, and for $s<k-1$
\begin{align}
    H(T_N) \leq \frac{k+(k-s-1)}{2} \log(2\pi e N)+ \frac{1}{2}\log\left( p_{1}\dots p_s \left(1- \sum_{j=1}^s p_j \right) \right) +\sum_{j=s+1}^k\frac{1}{2}\log( p_j^2 \frac{{\rm Var}(\tilde X^{(j)})}{h(\chi_j)^{2}})- o(1).
\end{align}
\end{proof}

By the lemma~\ref{lem: incomm properties} when the partition of $\chi$ in sets with single elements $\{\chi_j=\{x_j\}\}_{j=1,\dots d}$ is incommensurable,  so is the the partition obtained by merging two sets into one, e.g. $\chi_j=\{x_j\}$ for $j=1,\dots d-2$ and $\chi_{d-1}=\{x_{d-1},x_d\}$. One can verify that using the Theorem~\ref{Th:EntropyForIncomm} for these two partitions gives exactly the same bound for the entropy. Hence, in principle, one never needs to invoke the theorem for the case $s=k$.

\section{Incommensurability rank and the main result}

Upon further inspection one realises that the random variables for which $\chi$ admits an incommensurable partition in lattices are very special, hence the applicability of Theorem~\ref{Th:EntropyForIncomm} is quite limited. To obtain a bound applicable to all discrete random variables, we now introduce a notion which only characterizes the incommensurability of a subset of $\chi$.

\begin{defn}\label{def: k can}
    Consider a discrete set of  real numbers $\chi$. A collection of disjoint sets $\{\chi_{1},..\chi_{k}\}$ is called a \textbf{k-canonical prepartition} (or canonical prepartition) of $\chi$ if it satisfies the following properties: 
\begin{enumerate}
\item $\chi_j \subset\chi$ 
 \item $\chi_{j}$ is a non-empty lattice for $j=1,..,k$ (definition \ref{defn:lattice}).
 \item $\{\chi_{j}\}_{j=1,..,k} $ is incommensurable (definition \ref{def: incom}).
\end{enumerate}
If $|\chi_j| =1$ for all $j=1,\dots,k$ we say that the canonical prepartition is {\bf degenerate}. 
\end{defn}

Note that when $\bigcup_{j=1}^k \chi_j = \chi$, the collection $\{\chi_1\dots,\chi_k\}$ is an incommensurable partition of $\chi$ in lattices, recovering the setting of Theorem~\ref{Th:EntropyForIncomm}. 

For a general random variable characterizing all the canonical prepartitions is non-trivial (see the examples in the App.~\ref{IncommensurabilityExamples}). Nevertheless, the following general result is enlightening. To state it we recall the notion of the \textit{rational span} of a set $\chi\subset \mathbb{R}$, given by
\begin{equation}
{\rm span}_{\mathbb{Q}}(\chi):= \left \{ \underset{x \in \chi}{\sum} q_{x} x : q_{x} \in \mathbb{Q}\right\}.
\end{equation}
That is the set containing all linear combination of elements of $\chi$ with {\it rational} coefficients.

\begin{lem}\label{lem: part span}
 Any finite set $\chi=\{x_1,\dots,x_d\}$ admits a $q$-\textit{canonical prepartition}, where $q:={\rm dim}_{\mathbb{Q}} [{\rm span}_{\mathbb{Q}}(\chi)]$ is the dimension of the rational span of $\chi$ seen as  vector space over the field of rational numbers $\mathbb{Q}$. 
\end{lem}
\begin{proof}
First, reorder the values $(x_1,\dots,x_d)$ such that the first $q$ values $x_{1},x_{2},\dots,x_{q}$ are linearly independent and form a basis of the vector space ${\rm span}_{\mathbb{Q}}(\chi)$ (which implies that non of these elements is zero). Then, define the following sets 
\begin{equation}
    \chi_{j}:={\rm span}_{\mathbb{Q}}(x_{j})\cap (\chi\setminus\{0\}) \quad  \text{for} \quad j=1,\dots,q \qquad \text{and}  \qquad \overline \chi=\chi\setminus \underset{j=1}{\overset{q}{\bigcup}}\chi_j.
\end{equation} 
Linear independence of $x_{1},x_{2},\dots,x_{q}$ ensures that that $\{\chi_1,\dots, \chi_q\}$ are disjoint and incommensurable (implying 2), while $|\chi|=d$ with $x_j \in \chi_{j} \subset {\rm span}_{\mathbb{Q}}(x_{j})$ ensures that they are non-empty lattices with $\chi_j \neq \{0\}$ (implying 1 and 3). Hence $\{\chi_1,\dots, \chi_q,\overline \chi\}$ is a $q$-canonical prepartition of $\chi$.
\end{proof}

We are now ready to state the main result of the paper. But first we give a technical lemma used in the proof.

\begin{lem}\label{lemma: remove}
    Let $X_1,\dots, X_N$ be a collection of iid random variable taking values in some set $\chi$. For any bipartition $\chi =\chi_g \cup \chi_g^c$ in two disjoint sets $\chi_g \cap \chi_g^c= \varnothing$, with ${\rm Pr}(X_i \in \chi_g) = q$, define the random variables
    \begin{equation}\label{eq: Y lem}
        Y_i  \qquad \text{such that}\qquad {\rm Pr}(Y_i = x) = {\rm Pr}(X_i =x|X_i\in \chi_g)  = \frac{{\rm Pr}(X_i =x)}{q} 
    \end{equation}
   than the following bound holds
    \begin{equation}\label{eq: lemma 8}
        H\left(\sum_{i=1}^N X_i \right) \geq \sum_{n=0}^N  \binom{N}{n} q^n(1-q)^{N-n} \, H\left(\sum_{i=1}^n Y_i \right).
    \end{equation}
\end{lem}
\begin{proof}
    Each random variable $X_i$ can be decomposed as  $X_i= C_i Y_i   + (1-C_i) Z_i$ where $Y_i,Z_i,C_{i}$ 
    are independent, $C_i\in\{0,1\}$ is Bernoulli with ${\rm Pr}(C_i= 1) = q $, $Y_i$ is given in Eq.~\eqref{eq: Y lem}  and
    \begin{equation}\label{eq: Z lem}
        Z_i  \qquad \text{such that}\qquad {\rm Pr}(Z_i = x) = {\rm Pr}(X_i =x| X_i \in \chi_g^c)  = \frac{ {\rm Pr}(X_i =x)}{1-q}.
    \end{equation}
    This decomposition correspondences to first sampling a biased coin $C_i$ deciding if $X_i$ is in the set $\chi_g$ or $\chi_g^c$, and then sampling the value in these sets from the conditional distributions. We will use the notation $\bm C =C_1\dots C_N$, $\bm Z =Z_1\dots Z_N$  and $\lVert \bm c\rVert:=\sumab{i=1}{N}c_{i}$ counting the number of ones in a bitstrings (the value of $\bm C$). For these variables we have 
    \begin{align}
        H\left(\sum_{i=1}^N X_i \right) &= H\left(\sum_{i=1}^N (C_i Y_i  + (1-C_i) Z_i)\right) \\
        &\geq H\left(\sum_{i=1}^N (C_i Y_i  + (1-C_i) Z_i)\big | \bm C, \bm Z\right)  \\
        &=\sum_{\bm c, \bm z} Pr(\bm C=\bm c) Pr(\bm Z=\bm z) \,H \left(\sum_{i=1}^N c_i Y_i   + \sum_{i=1}^N  (1-c_i) z_i \right) \\
        &= \sum_{\bm c} Pr(\bm C=\bm c) \,\,H \left(\sum_{i=1}^N c_i Y_i \right) = \sum_{\bm c} Pr(\bm C=\bm c) \,\,H \left(\sum_{i=1}^{\lVert\bm c \rVert} Y_i \right) \\
        &= \sum_{n=0}^N  \binom{N}{n} q^n(1-q)^{N-n} \, H\left(\sum_{i=1}^n Y_i \right) .
    \end{align}
\end{proof}

\begin{thm}
 \label{Th:General-iid-entropy}
Let $X_1,\dots, X_N$ be a collection of iid random variable taking values in $\chi$, a set of real number that admits a k-canonical prepartition $\{\chi_{1},..\chi_{k}\}$ (definition \ref{def: k can}) with $p_j = {\rm Pr}(X_i \in \chi_j)$ for $j= 1,\dots,k$, and  $q =\sumab{j=1}{k} p_j$. In addition, let $\chi_{1},..\chi_{s\leq k}$ be the only subsets with one element, then the following bounds hold
 \begin{align}\label{eq: lambda1}
      s<k: \quad    H\left(\sum_{i=1}^N X_i \right) &\geq 
          \frac{k}{2} \log(2\pi e  N \lambda_1)-o(1) \qquad \quad \,\, \, \, \lambda_{1}:=\left( p_1\dots p_s \left(1- \sum_{j=1}^s \frac{p_j}{q}\right)\prod_{j=s+1}^k p_j \frac{{\rm Var}(\tilde X^{(j)})}{h(\chi_j)^{2}}\right )^{1/k} \\
    s = k: \quad   H\left(\sum_{i=1}^N X_i \right) &\geq  \frac{k-1}{2} \log(2\pi e  N \lambda_{2})-o(1)  \qquad \lambda_{2}:= \left (\frac{p_{1}\dots p_{k} }{q^{k}}\right )^{\frac{1}{k-1}}
 \end{align}

Where $\tilde X^{(j)}$ is the random variable such that ${\rm Pr}(\tilde X^{(j)}=x) ={\rm Pr}(X_{\rc{1}}=x| X_{\rc{1}} \in \chi_j)$ and $h(\chi_j)$ is the maximal span of the lattice $\chi_j$.
\end{thm}

\begin{proof}
    The proof simply consists of combining the Theorem~\ref{Th:EntropyForIncomm} with the Lemma~\ref{lemma: remove}. First, with the help of Lemma~\ref{lemma: remove} we get rid of all the outcomes in $\overline{\chi}:= \chi \setminus(\bigcup_{j=1}^k \chi_j)$ to obtain
    \begin{equation}
        H\left(\sum_{i=1}^N X_i \right) \geq \sum_{n=0}^N \binom{N}{n} q^n (1-q)^{N-n}  H\left(\sum_{i=1}^n Y_i \right),
    \end{equation}
where $Y_i$ are the conditional random variables 
\begin{equation}
Y_i \in \chi \setminus 
\overline{\chi} \qquad\quad {\rm Pr}( Y_i =x) = {\rm Pr}( X_i =x| X_i \notin \overline{\chi})  = \frac{{\rm Pr}( X_i =x) }{q}   \quad \text{for} \quad x \notin \overline{\chi}.
\end{equation}
Then we apply the Theorem~\ref{Th:EntropyForIncomm} to bound each term $H\left(\sumab{i=1}{n} Y_i \right)$.

For $s<k$ we find (using the fact that $\chi_j$ are lattices)
\begin{equation}
    H\left(\sum_{i=1}^n Y_i \right) \geq 
    \frac{k}{2} \log(2\pi e \, n)+ \frac{1}{2}\log\left( \frac{p_{1}}{q}\dots \frac{p_{s}}{q} \left(1-\sum_{j=1}^s \frac{p_j}{q} \right) \right) +\sum_{j=s+1}^k\frac{1}{2}\log( \frac{p_j}{q} \frac{{\rm Var}(\tilde X^{(j)})}{h(\chi_j)^{2}})- o(1),
\end{equation}
since $\tilde X^{(j)}=\tilde Y^{(j)}$. In the large $N$ limit we know that $\sumab{n=0}{N}\binom{N}{n} q^n (1-q)^{N-n} \log(n)=\log(qN)+o(1)$, leading to 
\begin{align}
     H\left(\sum_{i=1}^N X_i \right)  &\geq 
    \frac{k}{2} \log(2\pi e \, q N )+ \frac{1}{2}\log\left( \frac{p_{1}}{q}\dots \frac{p_{s}}{q} \left(1-\sum_{j=1}^s \frac{p_j}{q} \right) \right) +\sum_{j=s+1}^k\frac{1}{2}\log( \frac{p_j}{q} \frac{{\rm Var}(\tilde X^{(j)})}{h(\chi_j)^{2}})- o(1)\\
    & = \frac{k}{2} \log(2\pi e \, N) + \frac{1}{2}\log\left( p_1\dots p_s \prod_{j=s+1}^k p_j \frac{{\rm Var}(\tilde X^{(j)})}{h(\chi_j)^{2}}\right) + \frac{1}{2}\log\left(1-\sum_{j=1}^s \frac{p_j}{q} \right) -o(1).
\end{align}

For $s=k$ we get 
\begin{equation}
    H\left(\sum_{i=1}^n Y_i \right) \geq 
    \frac{k}{2} \log(2\pi e \, n)+ \frac{1}{2}\log\left( \frac{p_{1}}{q}\dots \frac{p_{k}}{q}\right) - o(1),
\end{equation}
 where once again we use that in the large $N$ limit $\mathds{E}[\log(n)]= \log (q N)+o(1)$, completing the proof.
 \end{proof}

Remark that if the sets $\chi_{1},..\chi_{k}$ all have a single element (i.e. the prepartition is degenerate), one can merge any two of them to obtain a $(k-1)$-canonical prepartition (non-degenerate). One can verify that with the above result, the two prepartitions lead to identical bounds. This observation, and the fact that the bound on entropy we obtained is linear in $k$ motivates the following definition.

\begin{defn} \label{Def:r-Set}
    Let $\chi$  be a discreet set of real numbers with $|\chi|\geq 2$.  We define its {\bf incommensurability rank} $r(\chi)$ as the maximal integer $1\leq r(\chi)\leq |\chi|-1$ such that $\chi$ admits a non-degenerate $r$-canonical prepartition.  For a singleton set $|\chi|=1$  we set $r(\chi)=0$.
\end{defn}

The lower bound $r(\chi)\geq 1$ follows from the observation that $\chi_1=\{x_1,x_2\}$ is a non-degenerate $1$-canonical prepartition for each $\chi$. While the upper-bound $r(\chi) \leq |\chi|-1$ trivially follows from the observation that a set can not be partitioned in more that $|\chi|-1$ disjoint subsets which are not all singleton.

\begin{defn} \label{Def:r-RV}
   Let $X$ be a discrete non-deterministic random variable $X$ taking values in $\chi\subset \mathds{R}$, we call its {\bf incommensurability rank} $r(X)$ the maximal integer $1\leq r(X)\leq |\chi|-1$ such that a shifted set $\chi' = \{x+a:x\in \chi\}$ with $a\in \mathds{R}$ admits a non-degenerate $r$-canonical prepartition.  For a deterministic random variables we set $r(X)=0$.
\end{defn}

Note that a priori, the incommensurability rank of a set $\chi$ can be increased by a shift of the set, see example around Eq.~\eqref{eq: example non invariant}. This is why we have split the definition of the rank for a set and for a random variable. Nevertheless, the Lemma~\ref{Lem:Conicidence} in the appendix shows that they always coincide, i.e. the incommensurability rank of a set $\chi$ can not be increased by a shift. Additionally, Lemma~\ref{lem: rank lattice} demonstrate that lattice random variables have unit rank $r(X)=1$, and, conversely, any random variable with finite domain $|\chi|<\infty$ and unit rank is lattice.\\

We can now state the equivalent of Theorem~\ref{Th:General-iid-entropy} in terms of the incommensurability rank of $X$, that is for the canonical prepartition with the optimal scaling of the entropy in the large $N$ limit.

\begin{cor} \label{cor:GeneralIIDEntropy}
Let $X_1,\dots, X_N$ be a collection of iid discrete random variable taking values in a set $\chi$ of real numbers. Then the following bounds holds
 \begin{align}
           H\left(\sum_{i=1}^N X_i \right) \geq \frac{r(\chi)}{2} \log( 2\pi e \lambda_1 N )-o(1)
    \end{align}
where $r(\chi)$ is the incommensurability rank of $X$ and $\lambda_1$ is given in Eq.~\eqref{eq: lambda1} for any non-degenerate $r(\chi)$-canonical partition of $\chi$.
\end{cor}

\section{conclusion}

While we do not posses explicit examples of random variables not saturating corollary~\ref{cor:GeneralIIDEntropy} 
in the scaling, we suspect that the lower bound is not tight in general. Instead, it would be interesting to establish an upper-bound in the general case, theorem \ref{Th:EntropyForIncomm} suggests that the entropy cannot scale faster than $H\leq \frac{2 r(\chi)-1}{2} \log(N) + {\it cst}$.\\

Finally, it is also worth mentioning that for a random variable $X$ admitting a $k$-canonical prepartition $\chi_g=\cup_{i=1}^k\chi_i \subset \chi$ with $|\chi_i|=1$ iff $i=1,\dots s$ and ${\rm Pr}(X_i \in \chi_g)=q$, the lower bound in theorem \ref{Th:General-iid-entropy} is obtained by combining the inequalities~(\ref{eq: lemma 8}),\eqref{eq: chanin rule} and \eqref{eq: entropy lb}. Those can be summarized as 
 \begin{align}
 H\left(\sum_{i=1}^N X_i \right) &\geq \sum_{n=0}^N \binom{N}{n} q^n (1-q)^{N-n}  \, H\left(\sum_{i=1}^n Y_i \right)\\
  H\left(\sum_{i=1}^n Y_i \right) &\geq H(n_n^{(1)},\dots,n_n^{(s)}) + H(Y_n^{(s+1)}|n_n^{(s+1)})+\dots+ H(Y_n^{(k)}|n_n^{(k)})
\end{align}
with the random variables $Y_i$ with ${\rm Pr}(Y_i = x) := {\rm Pr}(X_i =x|X_i\in \chi_g)$ and $n_n^{(j)} =\sumab{i=1}{n} \bm 1_{\chi_j}(Y_i)$. Both expressions are simple consequences of the data processing inequality for the conditional entropy (in addition to the incommensurability of the sets $\chi_i$), and hence directly apply to any entropy that has this property.\\

{\it Acknowledgements.---} We thank Bernardo Tarini and Pietro Gualdi for suggestions during the early stage of the project and L. Gavalakis with I. Kontoyiannis for encouraging comments on a draft of this manuscript. We acknowledge financial support from the Swiss National Science Foundation NCCR SwissMAP.

\appendix

\section{Examples of incommensurability partitions}  \label{IncommensurabilityExamples}

\subsection{Illustrative examples}
Let us illustrate the notion of incommensurability partitions with a pair of examples.
Consider the set $\chi = \bigcup_{j=1}^k \chi_j$ which is not necessarily finite. Let $\chi_j\subset \{ h_j n: n\in \mathbb{Z}\setminus\{0\}\}$ with $h_j>0$. If the real numbers $(h_1,\dots, h_k)$ are linearly independent over the field of rational numbers, i.e. $\sumab{j=1}{k} q_j h_j =0 \implies q_j=0$ for $q_j\in \mathbb{Q}$, then the partition $\{\chi_{j} \}_{j=1..k}$ is incommensurable. Indeed, linear independence guarantees that for integers $n_i$ the equation $t = n_1 h_1 +\dots n_k h_k$ has at most a unique solution.

Constructing partitions such that it's subsets are not linear independent (over $\mathbb{Q}$) is also straitforword, for example consider $\chi=\{1,\pi, 1+\pi\}$ into $\chi_{1}=\{1\},\chi_{2}=\{\pi \}$ and $\chi_{3}=\{1+\pi\}$. Indeed, the following system of equations for natural variables $n_1,n_2$ and $n_3$
\begin{equation}
    \left\{
\begin{array}{l}
n_{1}+n_{2}\pi + n_{3}(1+\pi) = t =k+k'\pi, \\
n_{1}+n_{2}+n_{3} = N.
\end{array}
\right. \iff  \left\{
\begin{array}{l}
n_{1}+n_{3} = k, \\
n_{2}+n_{3}=k', \\
n_{1}+n_{2}+n_{3} = N.
\end{array}
\right.
\end{equation}
has at most one solution. For this case, is crucial that in the definition of an incommensurable partition, the total number $N$ of sampled random variables is fixed.  

\subsection{An example of $r(\chi)>{\rm dim}_{\mathbb{Q}} [{\rm span}_{\mathbb{Q}}(\chi)]$}
Finally we make an example were the best $k$-canonical partition i.e. the one that yields the stronger inequality when applied to result \ref{Th:General-iid-entropy}, is not given by the identification $\chi_{j}={\rm span}_{\mathbb{Q}}(x_{j})\cap \chi$ suggested by lemma~\ref{lem: part span}.  Consider the set $\chi:=\{1, 2, \pi,1+\pi,\sqrt{2},2\sqrt{2}\}$ and the following two canonical partitions
\begin{enumerate}
    \item $\chi_{1}=\{1\},\;\chi_{2}=\{\pi\},\; \chi_{3}=\{\sqrt{2},2\sqrt{2}\}, \;\overline{\chi}=\{1+\pi\}$, is a non-degenerate $3$-canonical partition, obtained with the procedure indicated in lemma~\ref{lem: part span}. Accordingly ${\rm dim}_{\mathbb{Q}} [{\rm span}_{\mathbb{Q}}(\{1,\pi,1+\pi,\sqrt{2},2\sqrt{2}\})]=3$. By using theorem ~\ref{Th:General-iid-entropy} we would get an entropy lower bound scaling as $\frac{3}{2}\log{N}$.
    \item
    $\chi_{1}:=\{1\},\; \chi_{2}:= \{\pi\},\; \chi_{3}:= \{1+\pi\},\;\chi_{4}:= \{\sqrt{2},2\sqrt{2}\}$ is a non-degenerate $4$-canonical partition, yielding a lower bound on the entropy scaling as $2\log(N)$ by theorem ~\ref{Th:General-iid-entropy} or equivalently by Corollary \ref{cor:GeneralIIDEntropy}.
\end{enumerate}

 The partition 2. proves that $\chi$ has $r(\chi) \geq 4$, in fact, because the set has $5$ elements, we have $r(\chi) = 4$. Thus, the construction in Lemma~\ref{lem: part span} is here suboptimal. Furthermore, starting from this example, it is easy to construct sets for which
\[
r(\chi) - \dim_{\mathbb{Q}} [\operatorname{span}_{\mathbb{Q}}(\chi)]
\]
is arbitrarily large. Consider $2d+1$ real numbers $r, r_{1}, \ldots, r_{2d}$ that are linearly independent over $\mathbb{Q}$, and define the set
\[
\chi := \{r, 2r, r_{1}, r_{2}, r_{1}+r_{2}, r_{3}, r_{4}, r_{3}+r_{4}, \ldots, r_{2d-1}+r_{2d} \}.
\]
Then we have
\[
\dim_{\mathbb{Q}} [\operatorname{span}_{\mathbb{Q}}(\chi)] = \dim_{\mathbb{Q}} [\operatorname{span}_{\mathbb{Q}}(r,r_{1},..r_{2d})=2d+1,
\]
while the partition made by $\chi_{0} = \{r, 2r\}$, $\chi_{j} = \{r_{j}\}$ for $j \in \{1, \ldots, 2d\}$, and $\chi_{2d+k} := \{r_{2k-1} + r_{2k}\}$ for $k \in \{1, \ldots, d\}$ is non-degenerate and $(3d+1)$-incommensurable, proving that
\[
r(\chi) - \dim_{\mathbb{Q}} [\operatorname{span}_{\mathbb{Q}}(\chi)] = d.
\]

\subsection{Incommensurability is not invariant under uniform translations.}

Here, we justify why the maximization over shifts in definition \ref{Def:r-RV} was introduced. We show explicitly that the notion of incommensurable partition in not invariant under uniform translations. Consider two linearly independent real numbers $r_{1},r_{2}$ and $\chi=\{r_{1},2r_{1},r_{2},2r_{2}\}$ admits the incommensurable partition $\chi_1=\{r_{1},2r_{1}\}$ with $\chi_2=\{r_{2},2 r_{2}\}$, since $t= n\,r_{1} + m\, r_{2}$ determines the naturals $n$ and $m$ uniquely. In contrast, for the shifted set $\chi=\{r_{1}+\epsilon,2+\epsilon, r_{2}+\epsilon,2r_{2}+\epsilon\}$  the partition $\chi_1=\{r_{1}+\epsilon,2r_{1}+\epsilon\}$ with $\chi_2=\{r_{2}+\epsilon,2 r_{2} +\epsilon\}$ is no longer incommensurable. Indeed, for $N=3$ we can decompose $t=\sumab{i=1}{N}X_{i}$ in two different ways 
\begin{equation}\label{eq: example non invariant}
    t= 2r_{1}+2r_{2} +3 \epsilon = 2(r_{1}+\epsilon) +(2r_{2}+\epsilon) = (2r_{1}+\epsilon) +2(r_{2}+\epsilon)
\end{equation}
admits decompositions with different $\bm y = \binom{2(r_{1}+\epsilon)}{2r_{2}+\epsilon} \neq \bm y' =\binom{2r_{1}+\epsilon}{2(r_{2}+\epsilon)}$, where we used the notation introduced in \ref{def: incom}.

\section{Some properties of the incommensurability rank}

First, we list some properties of incommensurable partitions that follow immediately from the definition and will be used later:

\begin{lem}\label{lem: incomm properties}
Let $\{\chi_{j} \}_{j=1..k}$ be an incommensurable collection of disjoint sets of real numbers, then the following collections are also incommensurable
\begin{enumerate}
 \item Any collection where we remove one set $\{\chi_{j} \}_{j=1..k} \setminus \{\chi_i\} $.
 \item Any collection where we remove a subset from a set $(\{\chi_{j} \}_{j=1..k} \setminus \{\chi_i\}) \cup \{\chi_i^{sub} \} $, with $ \emptyset\neq \chi_i^{sub} \subset \chi_i$ and $\chi_i^{sub}\neq \{0\}$.
 \item Any collection where we merge two sets $(\{\chi_{j} \}_{j=1..k} \setminus \{\chi_i,\chi_{i'}\}) \cup \{\chi_i \cup \chi_{i'}\} $.
\end{enumerate}
\end{lem}
\begin{proof}
All the points follow as particular cases of the implication~\eqref{eq: incom def}. For the point 1, take the vectors $\bm y$ and $\bm y'$ with $y_i=m_i=y_i' =n_i =0$  and such that $\suma{j\neq i} (y_j-y_j')$, Eq.~\eqref{eq: incom def} implies $y_j = y_j'$ of all $j\neq i$ implying that the collection without the set $\chi_i$ is incommensurable.

For the point 2, take two vectors $\bm y$ and $\bm y'$ such that $y_i$ and $y_i'$ are obtained without using the values from $\chi_i^{sub}$. Again Eq.~\eqref{eq: incom def} implies $y_j = y_j'$ of all $j\neq i$.

For the point 3, for any pair of $k$-vector $\bm y$ and $\bm y'$ with $\suma{j}(y_j-y_j') =0$ define the $(k-1)$ vectors $\tilde {\bm y}$ and $\tilde {\bm y}'$ obtained by summing the values $\tilde y_i = y_i+y_{i'}$, with
$\suma{j}(y_j-y_j') =0 \Longleftrightarrow \suma{j}(\tilde y_j-\tilde y_j') =0$. Since the original collection is incommensurable $ \suma{j}(\tilde y_j-\tilde y_j') =0$ uniquely specifies the vectors  $\bm y$ with $\bm y'$ and $\tilde {\bm y}$ with $\tilde {\bm y}'$.
\end{proof}

We now prove that the incommensurability rank of a random variable can always be obtained without shifting it's outcomes.  

\begin{lem} \label{Lem:Conicidence} For a discrete random variable $X$ taking values in $\chi\subset\mathds{R}$ one has $r(\chi)=r(X)$.
\end{lem}
\begin{proof}
    By construction we know that $r(\chi)\leq r(X)$. Consider the shifted set $\chi':=\{x+c : x\in \chi\}$ such that $r(\chi')= r(X)$. It admits a canonical prepartition $\chi_1,\dots, \chi_r$. By lemma~\ref{lem: incomm properties} we can drop elements from these sets until obtaining the incommensurable prepartition
    \begin{equation}
        \chi_1' = \{x_0',x_1'\},\, \chi_2' =\{x_2'\},\, \dots \chi_r' =\{x_r'\}
    \end{equation}
with $x_i'\neq 0$ for $i\geq 2$. We thus have that for all $n_0+n_1+\dots n_r=N\in \mathbb{N}$ the sum 
\begin{align}
    t' = n_0 x_0' + n_1 x_1' + \dots n_r x_r'
\end{align}
implies a unique decomposition of $t'=\sumab{i=1}{r} y_r'$ with $y_1' =n_0 x_0' + n_1 x_1'$, $y_2 =n_2 x_2' \, \dots \,y_r' =n_r x_r'$. 
In turns, this implies that $n_2,\dots, n_r$, and thus $n_0+n_1 = N- \sumab{i=2}{r} n_i$ 
are uniquely specified by $t'$. But $y_1'  =n_0 x_0' + n_1 x_1'$ then also has a unique solution for 
$n_1$ and $n_2$. Hence the value of $t' = \sumab{i=0}{r} n_0 x_0' + n_1 x_1' + \dots n_r x_r'$ specifies all $n_i$ uniquely.

Now consider the shifted values $x_i = x_i' -a$ form the original set.  The observed total $t = \suma{i} n_i x_i$ must also determine all the $n_0,
\dots,n_r$, since it determines $t'= t + N a$, which determines all the $n_i$. Finally, to obtain a non-degenerate $r$-canonical prepartition from the singleton sets $\{x_0\}, \dots \{x_r\}$ simply group any two elements into a set with two elements. In particular, if $x_i=0$ it must one of the grouped values. This shows that if the shifted set $\chi'$ admits a $r$-canonical prepartition, so does the set $\chi$, completing the proof. 
\end{proof}

\begin{lem} \label{lem: rank lattice}
 A non-deterministic lattice random variable has rank $r(X)=1$, conversely any random variable $X$ with finite domain $|\chi|=d$ and rank $r(X)=1$ is a lattice random variable.
\end{lem}

\begin{proof} 
A deterministic random variable has $r(X)=0\neq1$. A random variable that takes two values is always lattice and has $r(X)=1$. Hence, we now assume $|\chi| \geq 3$.

We will first demonstrate that a non-deterministic lattice random variables $X$ has $r(X)=1$. We need to show that such $X$ admits no non-degenerate canonical $k$-partition with $k\geq 2$. By Lemma~\ref{lem: incomm properties} it is sufficient to show that any partition  $\chi_1 =\{x_1,x_2\}$ and $\chi_2=\{x_3\}$ and $\overline\chi=\chi\setminus \{x_1,x_2,x_3\}$ for any $x_1,x_2,x_3\in \chi$ is not $2$-canonical, since by applying the operations 1 and 2 from the Lemma any $k$-canonical non-degenerate partition can be reduced to such a $2$-canonical non-degenerate partition. Since $\chi$ is a lattice it is thus sufficient to show that the two set 
\begin{equation}
   \{a + h \,k_1,a + h \,k_2\} \qquad \text{and} \qquad \{a + h \,k_3\} \qquad \text{with} \qquad k_1,k_2,k_3 \in \mathbb{Z}, 
\end{equation}
are commensurable. To see this, for some $n_1',n_2',m_1',m_2'\in \mathbb{N}$, consider the two following vectors of natural numbers 
\begin{equation}
    \bm n =
    \left(\begin{array}{c}
    n_1=n_1' |k_2-k_3| \\
    n_2= n_2' |k_1-k_3| \\
    n_3=N - n_1 -n_2
    \end{array}\right), \quad
    \bm m = \left(\begin{array}{c}
    m_1=m_1' |k_2-k_3| \\
    m_2= m_2' |k_1-k_3| \\
    m_3=N - m_1 -m_2
    \end{array}\right) \in \mathbb{N}_N^3 
\end{equation}
where $N$ is chosen large enough such that $n_3,m_3\geq 0$. This choice defines the following real vectors
\begin{equation}
    \bm y = \binom{y_1}{y_2}=\binom{n_1 x_1 + n_2 x_2}{n_3 x_3} \quad \text{and}\qquad \bm y' = \binom{y_1'}{y_2'}=\binom{m_1 x_1 + m_2 x_2}{m_3 x_3}.
\end{equation}
By definition~\ref{def: incom} this choice exhibits the commensurability of the sets if $y_1+y_2 = y_1' +y_2'$ and $y_2 \neq y_2'$. We now show that there is always a choice of $n_1',n_2',m_1',m_2'$ fulfilling both conditions. The second condition can be rewritten as 
\begin{align} &y_2 \neq y_2' \quad \Longleftrightarrow \quad 
n_1 + n_2 \neq m_1 + m_2 \quad \Longleftrightarrow \quad 
(n_1'-m_1') |k_2-k_3| + (n_2'-m_2') |k_1-k_3| \neq 0. \label{eq: cond neq}
\end{align}
In turn, for the first one we get
\begin{align} &y_1+y_2 = y_1' +y_2' \qquad \Longleftrightarrow\quad n_1 x_1 + n_2 x_2 +(N-n_1-n_2) x_3 = m_1 x_1 + m_2 x_2 + (N-m_1-m_2) x_3 \qquad \Longleftrightarrow\\
&n_1 (x_1-x_3) + n_2 (x_2-x_3) = m_1 (x_1-x_3) + m_2 (x_2-x_3)  \qquad \Longleftrightarrow\\
    & n_1' |k_2-k_3|(k_1-k_3) + n_2' |k_1-k_3|(k_2-k_3) =
    m_1' |k_2-k_3|(k_1-k_3) + m_2' |k_1-k_3|(k_2-k_3) \qquad \Longleftrightarrow
    \\ 
    &(n_1' \pm n_2' ) |k_2-k_3|(k_1-k_3) =  (m_1' \pm m_2' ) |k_2-k_3|(k_1-k_3) \qquad \Longleftrightarrow \\
    &(n_1' \pm n_2' )  =  (m_1' \pm m_2' ), \label{eq: cond eq}
\end{align}
where $\pm= {\rm sign}(k_1-k_3)(k_2-k_3)$.

For ${\rm sign}(k_1-k_3)(k_2-k_3)=+1$ we know that $|k_2-k_3|\neq |k_1-k_3|$ and can choose $m_1' = n_1'+1$ with $m_2' = n_2'-1$ to satisfy both conditions \eqref{eq: cond neq} and \eqref{eq: cond eq}.

For ${\rm sign}(k_1-k_3)(k_2-k_3)=-1$ we can chose $m_1' = n_1'+1$ with $m_2' = n_2'+1$ to satisfy both conditions. Hence, the sets are always commensurable. \\

Next, we show the converse direction, $r(X)=1\implies \,X$ is lattice for all variable with a finite domain $|\chi|=d$. By assumption $r(X)=1$, we know that any partition of $\chi$ in $\{\chi_1, \dots, \chi_k,\overline \chi\}$, where $\chi_j$ are non-empty and $|\chi_1|\geq 1$ are not canonical. In particular,  any partition with $\chi_1=\{x_1,x_2\},  \chi_2=\{x_i\neq 0\}$ and $ \overline \chi = \chi\setminus \{x_1,x_2,x_i\}$ is not canonical. If $0 \in \chi$ we will chose $x_1=0$ and any other value for $x_2$, and otherwise pick both $x_1$ and $x_2$ arbitrarily, then we let $x_i$ run through the remaining (non-zero) elements of the set. Since the sets $\chi_1$ and $\chi_2$ are lattices, it follows that they are commensurable (not incommensurable).  
Hence, for each $x_i$ there exist $N$ and  $\bm n,\bm m \in \mathbb{N}_N^3$ such that 
\begin{align}
\begin{cases}
    n_1 x_1 + n_2 x_2 +n_3 x_i &= m_1 x_1 + m_2 x_2 + m_3 x_i  \\
    n_1+n_2+n_3 &= m_1+m_2+m_3 \\
n_3 x_i & \neq   m_3 x_i
\end{cases} \Longleftrightarrow
\begin{cases}
    x_i &= \frac{n_1-m_1}{m_3-n_3}x_1 + \frac{n_2-m_2}{m_3-n_3}x_2    \\
   m_3-n_3 &= (n_1-m_1) +(n_2-m_2) \\
m_3-n_3  & \neq  0 
\end{cases}
\end{align}
where we used the fact that $x_i \neq 0$. Defining the integers $k_3^{(i)}=m_3-n_3\neq 0$  with $k_2^{(i)}=n_2 -m_2$ and combining the first two equations we obtain the following identities
\begin{align}
    x_i &=  x_1 + \frac{k_2^{(i)}}{k_3^{(i)}}(x_2-x_1) \quad \forall \, i\geq 3 \\
    x_1 &= x_1 \\
    x_2 &= x_1 + (x_2-x_1)
\end{align}
Defining $h = \frac{(x_2-x_1)} {\prod_{j=3}^d k_3^{(j)}} \neq 0$ and $a=x_1$ we can rewrite these equations as
\begin{align}
    x_i &=  a + h \left( k_2^{(i)} \prod_{j=3, j\neq i}^d k_3^{(j)}\right ) \quad \forall \, i\geq 3 \\
    x_1 &= a +  h \cdot 0\\
    x_2 &= a + h \left(\prod_{i=3}^d k_3^{(j)}\right),
\end{align}
since $0, \left(\prod_{i=3}^d k_3^{(j)}\right)$ and $\left( k_2^{(i)} \prod_{j=3, j\neq i}^d k_3^{(j)}\right )$ are integers for all $i$, we conclude that $\chi$ is a lattice.
\end{proof}

It is worth mentioning that when $\chi$ is not finite, it can have unit rank $r(\chi)=1$ without being a lattice. For example, any set $\chi \subseteq \mathbb{Q}$ that contains elements arbitrarily close to 0 is not a lattice but has $r(\chi)=1$. \\

\bibliography{biblio}

\begin{thebibliography}{7}%
\makeatletter
\providecommand \@ifxundefined [1]{%
 \@ifx{#1\undefined}
}%
\providecommand \@ifnum [1]{%
 \ifnum #1\expandafter \@firstoftwo
 \else \expandafter \@secondoftwo
 \fi
}%
\providecommand \@ifx [1]{%
 \ifx #1\expandafter \@firstoftwo
 \else \expandafter \@secondoftwo
 \fi
}%
\providecommand \natexlab [1]{#1}%
\providecommand \enquote  [1]{``#1''}%
\providecommand \bibnamefont  [1]{#1}%
\providecommand \bibfnamefont [1]{#1}%
\providecommand \citenamefont [1]{#1}%
\providecommand \href@noop [0]{\@secondoftwo}%
\providecommand \href [0]{\begingroup \@sanitize@url \@href}%
\providecommand \@href[1]{\@@startlink{#1}\@@href}%
\providecommand \@@href[1]{\endgroup#1\@@endlink}%
\providecommand \@sanitize@url [0]{\catcode `\\12\catcode `\$12\catcode
  `\&12\catcode `\#12\catcode `\^12\catcode `\_12\catcode `\%12\relax}%
\providecommand \@@startlink[1]{}%
\providecommand \@@endlink[0]{}%
\providecommand \url  [0]{\begingroup\@sanitize@url \@url }%
\providecommand \@url [1]{\endgroup\@href {#1}{\urlprefix }}%
\providecommand \urlprefix  [0]{URL }%
\providecommand \Eprint [0]{\href }%
\providecommand \doibase [0]{http://dx.doi.org/}%
\providecommand \selectlanguage [0]{\@gobble}%
\providecommand \bibinfo  [0]{\@secondoftwo}%
\providecommand \bibfield  [0]{\@secondoftwo}%
\providecommand \translation [1]{[#1]}%
\providecommand \BibitemOpen [0]{}%
\providecommand \bibitemStop [0]{}%
\providecommand \bibitemNoStop [0]{.\EOS\space}%
\providecommand \EOS [0]{\spacefactor3000\relax}%
\providecommand \BibitemShut  [1]{\csname bibitem#1\endcsname}%
\let\auto@bib@innerbib\@empty
\bibitem [{\citenamefont {Linnik}(1959)}]{Linnik}%
  \BibitemOpen
  \bibfield  {author} {\bibinfo {author} {\bibfnamefont {J.~V.}\ \bibnamefont
  {Linnik}},\ }\href {https://doi.org/10.1137/1104028} {\bibfield  {journal}
  {\bibinfo  {journal} {Theory of Probability \& Its Applications}\ }\textbf
  {\bibinfo {volume} {4}},\ \bibinfo {pages} {288} (\bibinfo {year}
  {1959})}\BibitemShut {NoStop}%
\bibitem [{\citenamefont {Barron}(1986)}]{barron1986entropy}%
  \BibitemOpen
  \bibfield  {author} {\bibinfo {author} {\bibfnamefont {A.~R.}\ \bibnamefont
  {Barron}},\ }\href
  {https://www.jstor.org/stable/pdf/2244098.pdf?casa_token=7WV15i7DgpgAAAAA:JE078lBzFuXejcq3Rdn4wXJDr2-E45miOpQRwAFsZuAIeDZsAtVMc5O-T0YPNUCJ0jFQU720BejJBW447nWkTG-plPsmXoz3KLY3cAVfYquW8wtLhGYl}
  {\bibfield  {journal} {\bibinfo  {journal} {The Annals of probability}\
  }\textbf {\bibinfo {volume} {14}},\ \bibinfo {pages} {336} (\bibinfo {year}
  {1986})}\BibitemShut {NoStop}%
\bibitem [{\citenamefont {Artstein}\ \emph {et~al.}(2004)\citenamefont
  {Artstein}, \citenamefont {Ball}, \citenamefont {Barthe},\ and\ \citenamefont
  {Naor}}]{artstein2004solution}%
  \BibitemOpen
  \bibfield  {author} {\bibinfo {author} {\bibfnamefont {S.}~\bibnamefont
  {Artstein}}, \bibinfo {author} {\bibfnamefont {K.}~\bibnamefont {Ball}},
  \bibinfo {author} {\bibfnamefont {F.}~\bibnamefont {Barthe}}, \ and\ \bibinfo
  {author} {\bibfnamefont {A.}~\bibnamefont {Naor}},\ }\href
  {https://www.jstor.org/stable/pdf/20161222.pdf?casa_token=9q0ZUTjuissAAAAA:No9H9RqfL_BAVvSodol4abZax86odquKf_to6ibPHR5AaislkfmFOEhOfkf6-JRPrjLCopVk05JLsozGp5UFuG4IBQQrubySqeoMgHuBgOYQFrEfcdGj}
  {\bibfield  {journal} {\bibinfo  {journal} {Journal of the American
  Mathematical Society}\ }\textbf {\bibinfo {volume} {17}},\ \bibinfo {pages}
  {975} (\bibinfo {year} {2004})}\BibitemShut {NoStop}%
\bibitem [{\citenamefont {Madiman}\ and\ \citenamefont
  {Barron}(2007)}]{madiman2007generalized}%
  \BibitemOpen
  \bibfield  {author} {\bibinfo {author} {\bibfnamefont {M.}~\bibnamefont
  {Madiman}}\ and\ \bibinfo {author} {\bibfnamefont {A.}~\bibnamefont
  {Barron}},\ }\href
  {https://ieeexplore.ieee.org/iel5/18/4252316/04252338.pdf?casa_token=6GaIoetdaL4AAAAA:lofyRRHmMz3Qyh3EWCz7l105dAZ2PGJ5oRgA8Rm6odMga-u9xoxybtidoiKpBkIAzKIyzv9Tofe5dw}
  {\bibfield  {journal} {\bibinfo  {journal} {IEEE Transactions on Information
  Theory}\ }\textbf {\bibinfo {volume} {53}},\ \bibinfo {pages} {2317}
  (\bibinfo {year} {2007})}\BibitemShut {NoStop}%
\bibitem [{\citenamefont {Gavalakis}\ and\ \citenamefont
  {Kontoyiannis}(2024)}]{EntropySumIID}%
  \BibitemOpen
  \bibfield  {author} {\bibinfo {author} {\bibfnamefont {L.}~\bibnamefont
  {Gavalakis}}\ and\ \bibinfo {author} {\bibfnamefont {I.}~\bibnamefont
  {Kontoyiannis}},\ }\href {\doibase https://doi.org/10.1016/j.spa.2023.104294}
  {\bibfield  {journal} {\bibinfo  {journal} {Stochastic Processes and their
  Applications}\ }\textbf {\bibinfo {volume} {170}},\ \bibinfo {pages} {104294}
  (\bibinfo {year} {2024})}\BibitemShut {NoStop}%
\bibitem [{\citenamefont {Takano}(1987)}]{takano1987convergence}%
  \BibitemOpen
  \bibfield  {author} {\bibinfo {author} {\bibfnamefont {S.}~\bibnamefont
  {Takano}},\ }\href
  {https://ynu.repo.nii.ac.jp/record/6783/files/YMJ_35_N1-2_1987_143-148.pdf}
  {\bibfield  {journal} {\bibinfo  {journal} {Yokohama Mathematical Journal}\
  }\textbf {\bibinfo {volume} {35}},\ \bibinfo {pages} {143} (\bibinfo {year}
  {1987})}\BibitemShut {NoStop}%
\bibitem [{\citenamefont {Kaji}(2015)}]{MultinomialEntropy}%
  \BibitemOpen
  \bibfield  {author} {\bibinfo {author} {\bibfnamefont {Y.}~\bibnamefont
  {Kaji}},\ }in\ \href
  {https://ieeexplore.ieee.org/iel7/7270906/7282397/07282678.pdf?casa_token=TiguLeCRg7gAAAAA:S5Lc4m4ZUlg6aIjz9_RprDEgqghz_u7H2MEKC7v1FsILfnoA7VoGOwgTzgzvf7LsJ5yDiL-qJAR72w}
  {\emph {\bibinfo {booktitle} {2015 IEEE International Symposium on
  Information Theory (ISIT)}}}\ (\bibinfo {organization} {IEEE},\ \bibinfo
  {year} {2015})\ pp.\ \bibinfo {pages} {1362--1366}\BibitemShut {NoStop}%
\end{thebibliography}%
\end{document}